\documentclass[%
 reprint,
superscriptaddress,
 amsmath,amssymb,
 aps,
pra
]{revtex4-2}

\usepackage{graphicx}
\usepackage{dcolumn}
\usepackage{bm}
\usepackage{comment} 
\usepackage{hyperref}
\usepackage{mathtools}
\usepackage{braket}

\usepackage{tikz-cd}

\usepackage{siunitx}

\usepackage{amsthm}
\newtheorem{theorem}{Theorem}
\newtheorem*{theorem*}{Theorem}
\newtheorem*{lemma*}{Lemma}

\theoremstyle{definition}
\newtheorem{definition}{Definition}
\newtheorem{property}{Property}
\newtheorem{question}{Question}
\newtheorem*{notation}{Notation}

\theoremstyle{remark}
\newtheorem*{remark}{Remark}

\usepackage{xcolor}

\newcommand{\red}[1]{\textcolor{black}{#1}}

\newcommand{\rred}[1]{\textcolor{black}{#1}}

\usepackage{etoolbox}

\makeatletter
\patchcmd{\frontmatter@abstract@produce}
  {\vskip200\p@\@plus1fil
   \penalty-200\relax
   \vskip-200\p@\@plus-1fil}
  {}
  {}
  {}
\makeatother

\begin{document}


\title{Topological Winding Guaranteed Coherent Orthogonal Scattering}

\author{Cheng Guo}
\email{guocheng@stanford.edu}
\affiliation{
Ginzton Laboratory and Department of Electrical Engineering, Stanford University, Stanford, California 94305, USA
}

\author{Shanhui Fan}
\email{shanhui@stanford.edu}
\affiliation{
Ginzton Laboratory and Department of Electrical Engineering, Stanford University, Stanford, California 94305, USA
}%

\date{\today}

\begin{abstract}
Coherent control has enabled various novel phenomena in wave scattering. We introduce an effect called coherent orthogonal scattering, where the output wave becomes orthogonal to \rred{the reference output state without scatterers. This effect leads to a unity extinction coefficient and complete mode conversion.} We examine the conditions for this effect and reveal its topological nature by relating it to the indivisibility between the dimension and the winding number of scattering submatrices. These findings deepen our understanding of topological scattering phenomena.
\end{abstract}

\maketitle


\red{Understanding wave scattering is crucial for various applications, such as  imaging~\cite{feng1988a,vellekoop2007b,popoff2010,popoff2010c,vellekoop2010a,choi2011a,popoff2011a,bertolotti2012a,mosk2012c,chaigne2014a,katz2014a,kim2015e}, sensing~\cite{muraviev2018,tan2020}, energy harvesting~\cite{chen2005,zhang2007,howell2016,fan2017,cuevas2018b,ottens2011, raman2014, boriskina2016,raj2017a,fiorino2018a,park2021,zhu2019c,li2019g,li2021e,liu2022a,guo2022b,guo2023b}, and optical computing~\cite{miller2013b,carolan2015,shen2017,zhu2021,long2021}. Coherent control of wave scattering~\cite{popoff2014,liew2016,mounaix2016} has been a significant advancement, enabling the tailoring of scattering behaviors by modifying the input wave profile. This approach has unveiled unique phenomena, particularly coherent perfect absorption~\cite{chong2010a,wan2011,sun2014,baranov2017,mullers2018,pichler2019,sweeney2019,chen2020f,wang2021,slobodkin2022}--the complete absorption of a tailored input wave profile--and related effects~\cite{krasnok2019a,ni2023,guo2023}, including coherent virtual absorption~\cite{longhi2018,trainiti2019,zhong2020b} and reflectionless scattering modes~\cite{sweeney2020a,stone2021,horodynski2022,sol2023}.}

\begin{figure}[htbp]
    \centering
    \includegraphics[width=0.45\textwidth]{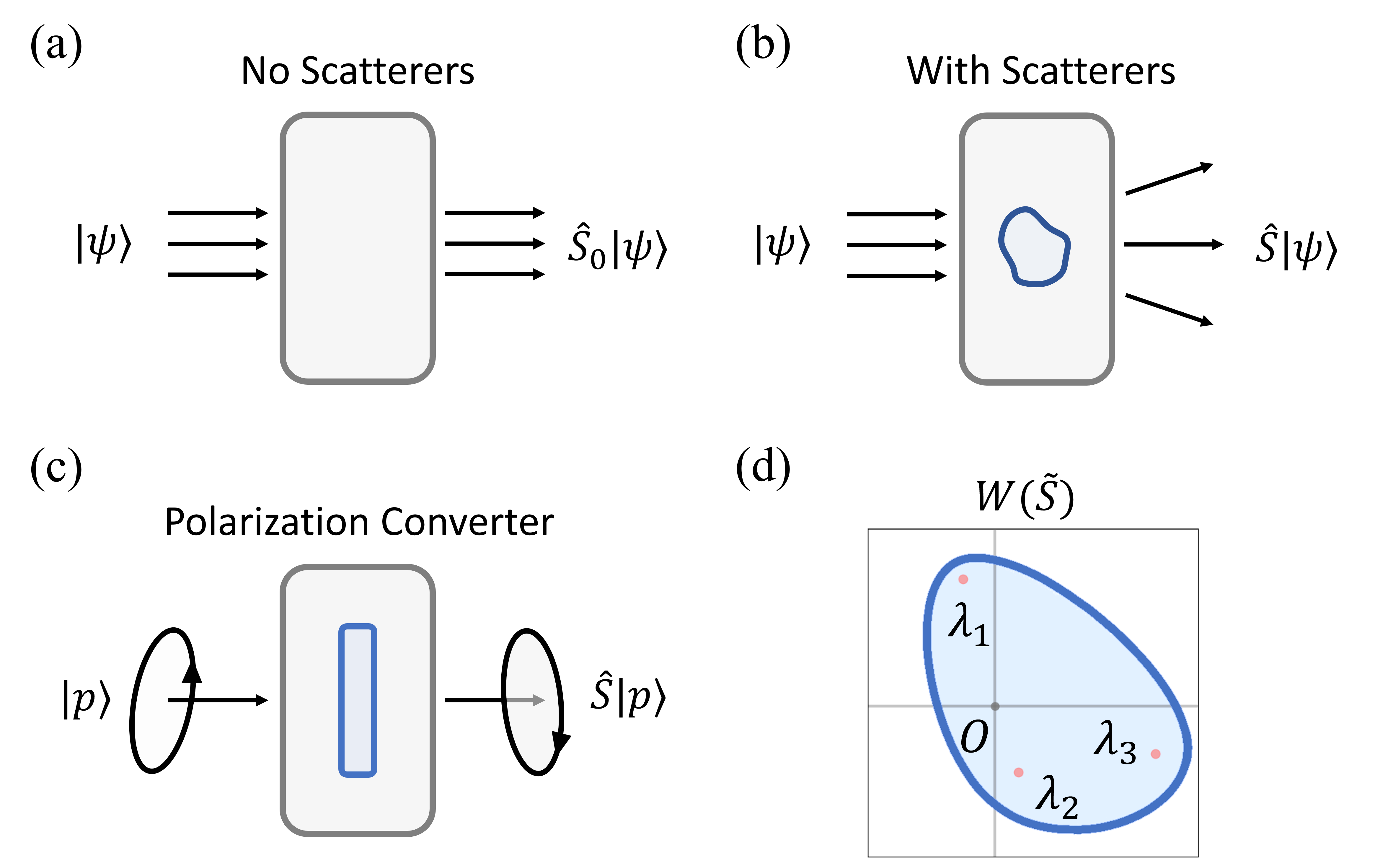}
    \caption{\rred{The concept of coherent orthogonal scattering (COS). (a,b) A standard scattering setup with an input state $\ket{\psi}$. (a) In the absence of scatterers, the output state is $\ket{\xi_0} = \hat{S}_0 \ket{\psi}$. (b) In the presence of scatterers, the output state becomes $\ket{\xi} = \hat{S} \ket{\psi}$. COS occurs when $\braket{\xi_0|\xi} = 0$. (c) In polarization optics, COS corresponds to complete polarization conversion. (d) The numerical range $W(\tilde{S})$ of a scattering submatrix $\tilde{S}$ is convex and contains its eigenvalues. COS occurs iff $0 \in W(\tilde{S})$.}} 
    \label{fig:scheme_COS}
\end{figure}

\rred{In this Letter, we introduce a new phenomenon called \emph{coherent orthogonal scattering}, which refers to the complete extinction of the output state with respect to the reference output state in the absence of scatterers. Consider a standard scattering setup~\cite{bohren2004,taylor2006}. Let $H$ be the Hilbert space of waves. We input a normalized state $\ket{\psi}$ into a system. In the absence of scatterers (Fig.~\ref{fig:scheme_COS}a), the output state, termed as the reference state, is 
\begin{equation}
\ket{\xi_{0}} = \hat{S}_{0} \ket{\psi},    
\end{equation}
where $\hat{S}_{0}$ is a unitary background scattering operator. In the presence of scatterers (Fig.~\ref{fig:scheme_COS}b), the output state is
\begin{equation}
\ket{\xi} = \hat{S} \ket{\psi},    
\end{equation}
where $\hat{S}$ is the total scattering operator which can be non-unitary. By orthogonal projection with respect to $\ket{\xi_{0}}$, the output state $\ket{\xi}$ can be uniquely decomposed as 
\begin{equation}
\ket{\xi} = \alpha \ket{\xi_{0}} + \beta  \ket{\xi_{0}^\bot},     
\end{equation}
where $\ket{\xi_0^\bot}$ is a unit vector orthogonal to $\ket{\xi_0}$. The absorption coefficient is defined as
\begin{equation}
\chi_{a} = 1 - \lvert \alpha \rvert^{2} - \lvert \beta \rvert^{2}  =  1 - \braket{\psi|\hat{S}^{\dagger}\hat{S}|\psi}.  
\end{equation}
The extinction coefficient is defined as 
\begin{equation}
\chi_{e} = 1 - \lvert \alpha \rvert^{2} =  1 - \lvert \braket{\xi_0|\xi}\rvert^{2} =
1 - \lvert \braket{\psi|\hat{S}_0^\dagger \hat{S}|\psi} \rvert^{2},  
\end{equation}
which measures the deviation of $\ket{\xi}$ from $\ket{\xi_0}$. \footnote{\rred{The scattering coefficient is defined as
$\chi_{s} = \lvert \beta \rvert^{2} = \chi_{e} - \chi_{a}$.} }
For passive scatterers, $\chi_a$ and $\chi_e$ lie between $0$ and $1$.} 

\rred{In coherent control experiments, one varies $\ket{\psi}$ within an accessible input subspace $H_i \subseteq H$ to achieve desired outcomes~\cite{ambichl2017,rotter2017,cao2022a}. Coherent perfect absorption refers to the case when $\chi_a = 1$. 
In contrast, coherent orthogonal scattering occurs when $\chi_e = 1$, which is equivalent to the output state being orthogonal to the reference state:
\begin{equation}\label{eq:COS-braket} 
\braket{\xi_0|\xi} = \braket{\psi|\hat{S}_0^\dagger\hat{S}|\psi} = 0.    
\end{equation}
For example, in polarization optics, $H_i$ is the two-dimensional ($2$D) polarization space, and the coherent orthogonal scattering corresponds to complete polarization conversion~(Fig.~\ref{fig:scheme_COS}c). A related yet distinct effect, coherent perfect extinction~\cite{guo2023},  refers to the case when the output state is orthogonal to the reference output subspace $H_o = \hat{S}_0 H_i$:
\begin{equation}\label{eq:CPE-braket} 
\forall \ket{\phi} \in H_i, \quad \braket{\phi|\hat{S}_0^\dagger\hat{S}|\psi} = 0.
\end{equation}
We can express Eqs.~(\ref{eq:COS-braket}) and (\ref{eq:CPE-braket}) in matrix forms. 
Suppose $H_i$ is $m$-dimensional. Using a set of orthonormal bases, the compression of $\hat{S}_0^\dagger\hat{S}$ to $H_i$ is represented by an $m\times m$ scattering submatrix $\tilde{S}$~\footnote{\rred{The compression of a linear operator $\hat{O}$ on a Hilbert space $H$ to a subspace $H_i$ is the operator $\hat{P} \hat{O}: H_i \to H_i$, where $\hat{P}$ is the projection from $H$ to $H_i$~\cite{halmos1982,istratescu2020}}}. (In the absence of scatterers, $\hat{S} = \hat{S}_0$, thus $\tilde{S} = I$.) Then Eq.~(\ref{eq:COS-braket}) becomes
\begin{equation}\label{eq:COS_definition}
\exists \, \tilde{\bm{a}} \in \mathbb{C}^m \backslash \{\bm{0}\}, \quad \tilde{\bm{a}}^{\dagger} \tilde{S} \tilde{\bm{a}} = 0,  
\end{equation}
while Eq.~(\ref{eq:CPE-braket}) becomes}
\begin{equation}\label{eq:CPE_definition}
\rred{\exists \, \tilde{\bm{a}} \in \mathbb{C}^m \backslash \{\bm{0}\}, \quad  \tilde{S} \tilde{\bm{a}} = \bm{0}.}  
\end{equation}

\rred{The physical significance of coherent orthogonal scattering is two-fold. First, it corresponds to a unity extinction coefficient, a fundamental observable in scattering experiments with applications in diverse areas, including molecular spectroscopy~\cite{perkampus1992,hammes2005,parson2015,linne2024}, acoustic imaging~\cite{hoff2001}, atmospheric science~\cite{platt2008,rao2012}, and astronomy~\cite{krasnok2019a,tennyson2019}. Understanding the conditions for achieving a unity extinction coefficient is important. Second, coherent orthogonal scattering enables complete mode conversion~\cite{hill1990,lee2001a,miller2012,lu2012,wang2012b,heinrich2014a,ohana2014,zhao2016d}, such as complete polarization conversion~\cite{huan2000,guo2017,guo2019c,guo2020b,gutierrez-vega2020,zeng2021a,liu2021g,biswas2021,dorrah2021a,he2021a,li2021g,li2022a,kang2022}, which plays a crucial role in various applications, including communications~\cite{kim2009a,wang2012b,li2015,fu2018}, sensing~\cite{vukusic1992,gu2014a,verre2016,zhang2021h,zhang2021g}, and quantum technology~\cite{xavier2008,mohanty2017a,fabre2020b,liang2023}. A key challenge in mode converter design is determining whether complete mode conversion can be achieved by adjusting specific design parameters, and efficiently identifying the parameter values that result in complete mode conversion~\cite{li2015, guo2017,guo2019a,dorrah2021a}. This again requires a deeper understanding of the conditions for coherent orthogonal scattering. }

\rred{In this work, we investigate the conditions for coherent orthogonal scattering. A useful mathematical concept for this purpose is the \emph{numerical range} of  $\tilde{S}$~\cite{horn2008,gustafson1997}:}
\begin{equation}\label{eq:def_numerical_range}
W(\tilde{S}) \coloneqq \{\bm{x}^{\dagger} \tilde{S} \bm{x}: \bm{x} \in \mathbb{C}^m, \bm{x}^{\dagger}\bm{x} = 1\}. 
\end{equation}
$W(\tilde{S})$ is a compact \emph{convex} subset of $\mathbb{C}$ that contains all the eigenvalues of $\tilde{S}$~\cite{toeplitz1918,hausdorff1919}. (See Supplemental Material (SM)~\footnote{See Supplemental Material at [URL] for numerical range, examples of scattering submatrices, Kato's theorem, matrix determinant lemma, proof of Eq.~(\ref{eq:main_result}), topological phases of a resonance, and additional numerical results, which includes 
Refs.~\cite{munkres2000,kippenhahn1951a,zachlin2008,helton2007,henrion2010,helton2012,keeler1997,brown2004,gau2006,militzer2017,cox2021,geryba2021,meyer2000,needham2021,zhao2019c}.}, Sec.~\ref{SI-sec:numerical_range} for more details.) Figure~\ref{fig:scheme_COS}d shows $W(\tilde{S})$ for an $\tilde{S}\in M_3$.  
The \emph{Crawford number} of $\tilde{S}$ is the distance of $W(\tilde{S})$ from the origin~\cite{crawford1976,wang2010,choi2023}:
\begin{equation}\label{eq:def_inner_numerical_radius}
c(\tilde{S}) \coloneqq \min \{|z|: z \in W(\tilde{S}) \}.   
\end{equation}
\red{From Eq.~(\ref{eq:COS_definition}), coherent orthogonal scattering occurs iff} 
\begin{equation}\label{eq:inclusion_relation}
0 \in W(\tilde{S}), \quad \text{i.e.,} \quad  c(\tilde{S}) = 0.  
\end{equation}
The condition~(\ref{eq:inclusion_relation}) can be numerically checked using an approximate algorithm~\cite{horn2008} (see SM, Sec.~\ref{SI-subsec:numerical_algorithm}), but it cannot be determined analytically using the entries of $\tilde{S}$ when $m\geq 4$~\cite{psarrakos2002,wu2019h}. In contrast, coherent perfect extinction occurs iff~\cite{guo2023} 
\begin{equation}
\det \tilde{S} = 0, \label{eq:det_S_equal_0}
\end{equation} 
which can be easily checked using the entries of $\tilde{S}$. \footnote{This difference arises because coherent orthogonal scattering, as defined by a system of quadratic equations [Eq.~(\ref{eq:COS_definition})] is more complicated than coherent perfect extinction, as defined by a system of linear equations [Eq.~(\ref{eq:CPE_definition})].} \rred{Hence, the key challenge is to find a simple analytical criterion for coherent orthogonal scattering.}

\begin{figure}[htbp]
    \centering
    \includegraphics[width=0.40\textwidth]{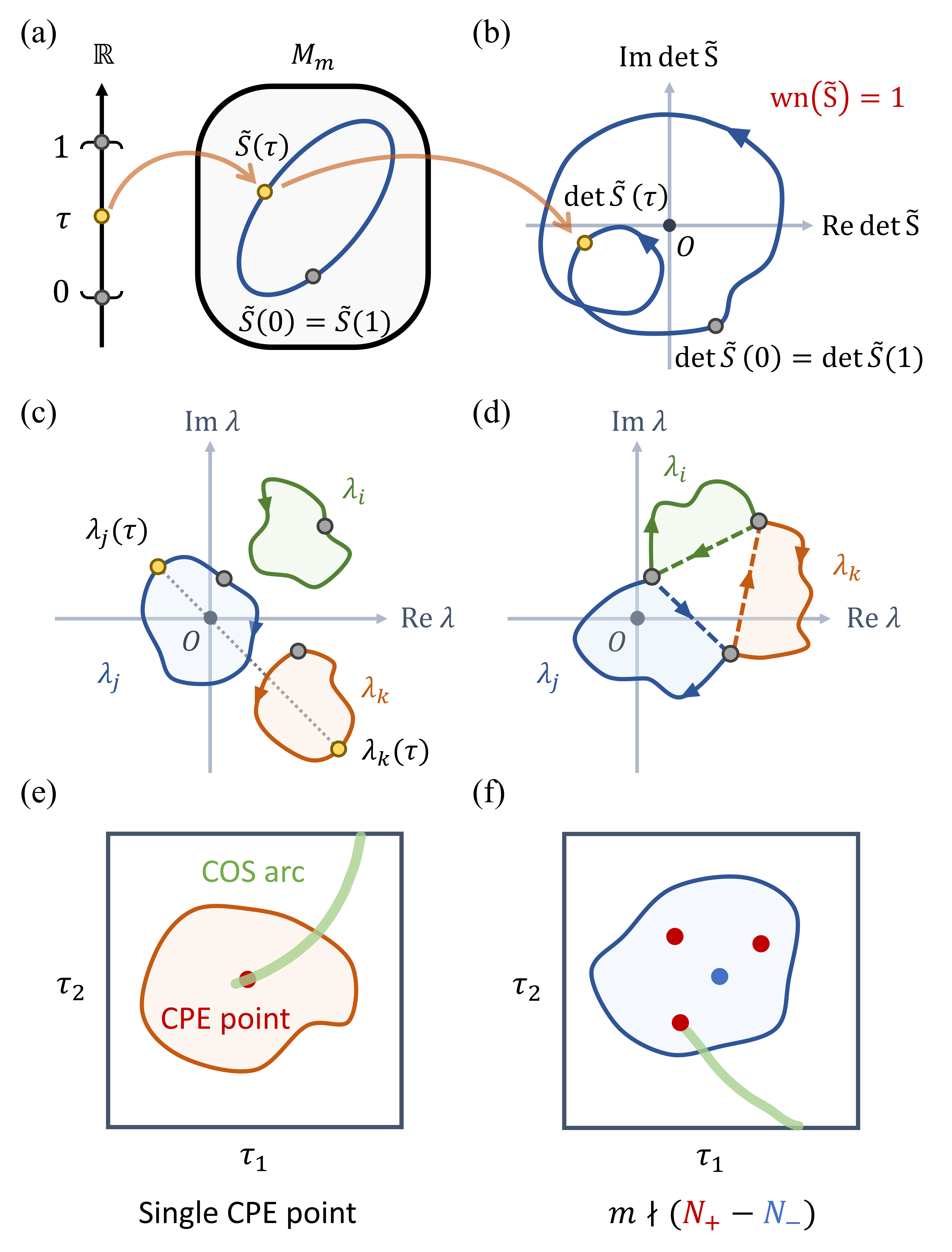}
    \caption{Topological winding guaranteed coherent orthogonal scattering. (a) A loop of $\tilde{S}(\tau) \in M_m$. (b) $\det \tilde{S}(\tau)$ has a winding number $\operatorname{wn}(\tilde{S})$. If $m$ does not divide $\operatorname{wn}(\tilde{S})$, denoted as $m \nmid \operatorname{wn}(\tilde{S})$, coherent orthogonal scattering must occur for some $\tilde{S}(\tau_0)$. (c,d) The path of eigenvalues when (c) $\red{\nu}(\tilde{S})=\operatorname{id}$, (d) $\red{\nu}(\tilde{S})\neq \operatorname{id}$. (e, f) CPE points and COS arcs in a 2D parameter space. (e) A topological COS arc exists associated with each CPE point. (f) A topological COS arc exists across a simple loop that surrounds multiple CPE points if $m \nmid (N_+ - N_-)$. COS arcs within the loop are not shown. }
    \label{fig:topological_winding}
\end{figure}

Here, we present a simple sufficient analytical criterion for coherent orthogonal scattering\rred{. Our idea is to examine} not a single, but a \emph{loop} of scattering submatrices (Fig.~\ref{fig:topological_winding}a). 
Consider a continuous map
\begin{align}
\tilde{S}: \; [0,1] \to M_{m}, \quad 
\tau \mapsto \tilde{S}(\tau),
\end{align}
where $ \tilde{S}(0) = \tilde{S}(1)$.
We ask whether 
\begin{equation}
\exists \, \tau_0 \in [0,1], \quad   0 \in W(\tilde{S}(\tau_0)),
\end{equation} 
\rred{motivated by the question of whether coherent orthogonal scattering can be achieved by tuning a certain parameter.}

Our criterion is as follows. We calculate $\det \tilde{S}(\tau)$ and check whether 
\begin{equation}
\exists \, \tau_0 \in [0,1], \quad \det \tilde{S}(\tau_0) = 0.    
\end{equation}
If true, then $\tilde{S}(\tau_{0})$ exhibits coherent perfect extinction, and thus also coherent orthogonal scattering. If false, then $\det \tilde{S}$ traces a closed path in $\mathbb{C}\backslash\{0\}$ (Fig.~\ref{fig:topological_winding}b), which has a well-defined winding number around the origin~\cite{roe2015,guo2023}:
\begin{equation}
\operatorname{wn}(\tilde{S})  \coloneqq \varphi(1) - \varphi(0) \in \mathbb{Z},    
\end{equation}
where $2\pi \varphi(\tau)$ is the continuous polar angle of $\det \tilde{S}(\tau)$. 
\red{If $m$ does not divide $\operatorname{wn}(\tilde{S})$, denoted as $m \nmid \operatorname{wn}(\tilde{S})$, then}
\begin{equation} \label{eq:criterion_to_prove}
\exists \, \tau_0 \in [0,1], \quad 0 \in W(\tilde{S}(\tau_0)).  
\end{equation} 
\red{If $m$ divides $\operatorname{wn}(\tilde{S})$, denoted as $m \mid \operatorname{wn}(\tilde{S})$~\cite{stein2009}, such $\tau_{0}$ may or may not exist.} See SM, Sec.~\ref{SI-sec:examples_scattering_submatrices} for examples.

The criterion~(\ref{eq:criterion_to_prove}) is our main result. It can be proved using a theorem established in Refs.~\cite{spitkovsky1974,guo2023i}. Here, we outline the essential ideas. Consider the $m$ eigenvalues of $\tilde{S}(\tau)$, denoted as $\lambda_{1}(\tau)$, ..., $\lambda_{m}(\tau)$. Since $\tilde{S}(\tau)$ is a continuous map from $[0,1]$ to $\operatorname{GL}_n$, we can choose all $\lambda_{j}(\tau)$'s to be continuous functions from $[0,1]$ to $\mathbb{C}\backslash\{0\}$~\cite{kato1995,li2019a} (see SM, Sec.~\ref{SI-Sec:Kato}.). Since $\tilde{S}(0)=\tilde{S}(1)$, the set of $\lambda_{j}(1)$'s must coincide with the $\lambda_{j}(0)$'s, up to a permutation:
\begin{equation}
\lambda_{j}(1) = \lambda_{\red{\nu}_{j}}(0); \quad 
\red{\nu} (\tilde{S}) = \begin{pmatrix}
1 & 2 &\dots & m \\
\red{\nu}_{1} & \red{\nu}_{2} &\dots &\red{\nu}_{m}
\end{pmatrix}. 
\end{equation}
For simplicity, let us first consider the special case when the permutation is the identity ($\red{\nu} (\tilde{S}) = \operatorname{id}$). Then, each $\lambda_j(\tau)$ traces out a loop in the punctuated complex plane with a well-defined winding number about the origin:
\begin{equation}\label{eq:def_wind_lambda_j}
\operatorname{wn}(\lambda_{j}) \in \mathbb{Z}, \qquad j=1,\dots,m.    
\end{equation}
There is a simple yet important relation:
\begin{equation}\label{eq:wn_S_and_wn_lambda}
\operatorname{wn}(\tilde{S}) = \sum_{j=1}^{m} \operatorname{wn}(\lambda_{j}),    
\end{equation}
which follows from the fact that~\cite{roe2015} 
\begin{equation}
    \det \tilde{S}(\tau) = \prod_{i=1}^{m} \lambda_{j}(\tau).
\end{equation}
Combining Eq.~(\ref{eq:wn_S_and_wn_lambda}) with our premise $m\nmid \operatorname{wn}(\tilde{S})$, there must exist a pair of eigenvalue loops, say $\lambda_{j}$ and $\lambda_{k}$, that have different winding numbers. As $\tau$ runs over $[0,1]$, the line segment $\overline{\lambda_{j}(\tau)\lambda_{k}(\tau)}$ connecting these two eigenvalues must sweep across the origin at least once, implying that 
\begin{equation}
\exists \tau_0 \in [0,1], \quad 0 \in \overline{\lambda_{j}(\tau_0)\lambda_{k}(\tau_0)} \subseteq W(\tilde{S}(\tau_{0})).   
\end{equation}
Here we used the fact that the numerical range contains the convex hull of the eigenvalues~\cite{toeplitz1918,hausdorff1919}. This completes the proof for the $\red{\nu}(\tilde{S})=\operatorname{id}$ case. To prove the $\red{\nu}(\tilde{S}) \neq \operatorname{id}$ cases (Fig.~\ref{fig:topological_winding}d), we construct $\tilde{S}': [0,1] \to \operatorname{GL}_{m}$, which undoes the permutation $\red{\nu}(\tilde{S})$ along the line segments $\overline{\lambda_{\red{\nu}_j}(0)\lambda_{j}(0)}$. The concatenation of $\tilde{S}$ and $\tilde{S}'$ becomes a loop with the identity permutation, allowing the previous analysis to complete the proof~\cite{guo2023i}.

\rred{With the mathematical groundwork established,} we now discuss the physical implications of criterion~(\ref{eq:criterion_to_prove}). First, our criterion establishes a deeper connection between coherent orthogonal scattering and coherent perfect extinction. Consider an $\tilde{S} \in M_m$ with $m\ge 2$ that depends continuously on two parameters $\tau = (\tau_{1},\tau_{2}) \in \Omega$, where $\Omega$ is a compact and simply connected subset of $\mathbb{R}^2$. Coherent perfect extinction generically occurs in $\Omega$ at isolated points, known as CPE points~\cite{guo2023}. Along a simple closed curve that encloses a single generic CPE point (Fig.~\ref{fig:topological_winding}e),
\begin{equation}
\operatorname{wn}(\tilde{S})=\pm 1.    
\end{equation}
Since $m \nmid \pm 1$, according to~(\ref{eq:criterion_to_prove}), coherent orthogonal scattering must occur somewhere along the loop. We can deform the loop continuously and deduce that coherent orthogonal scattering must occur along at least one arc, referred to as a COS arc. Such a COS arc is topologically protected. We can extend this analysis to the case of multiple CPE points (Fig.~\ref{fig:topological_winding}f). Suppose there are $N_{+}$ and $N_{-}$ CPE points with winding numbers $+1$ and $-1$, respectively, enclosed by a simple closed curve. A COS arc must exist across the loop if
\begin{equation}
m \nmid (N_{+}-N_{-}).    
\end{equation}

Second, we investigate the topological winding of scattering submatrices due to a resonance. We consider a single-resonance scattering submatrix~\cite{haus1984,fan2003,suh2004}
\begin{equation}\label{eq:S_tilde_sub_matrix}
\tilde{S}(\omega) = \tilde{C} + \frac{\tilde{\bm{d}} \tilde{\bm{\kappa}}^T}{-i(\omega-\omega_{0})+\gamma} \in M_m,  
\end{equation}
where $\omega_0$ and $\gamma$ are the resonant frequency and decay rate, respectively. The column vectors $\tilde{\bm{\kappa}}$ and $\tilde{\bm{d}}$ represent the coupling rates between the resonator and the input and output waves in the ports, respectively. $\tilde{C} \in M_m$ describes the background scattering. We can prove that $\det \tilde{S}(\omega)$ traces out a circle in the complex plane with
\begin{equation}
\label{eq:main_result}
\operatorname{wn}(\tilde{S}) =  \begin{cases}
0,  & \operatorname{Re} \tilde{\rho} < \frac{1}{2} \\
1, & \operatorname{Re} \tilde{\rho} > \frac{1}{2}
\end{cases},    \quad 
\tilde{\rho} \coloneqq -\frac{\tilde{\bm{\kappa}}^T \tilde{C}^{-1} \tilde{\bm{d}}}{2 \gamma} \in \mathbb{C}. 
\end{equation}
When $\operatorname{Re} \tilde{\rho} = \frac{1}{2}$, the circle passes through the origin, and $\operatorname{wn}(\tilde{S})$ is undefined. (See the proof in SM, Sec.~\ref{SI-sec:matrix_determinant_lemma}-\ref{SI-sec:topological_phases_resonance}.) Combining (\ref{eq:main_result}) and (\ref{eq:criterion_to_prove}), we conclude that when $m \ge 2$, 
\begin{equation}\label{eq:condition_resonance}
    \operatorname{Re} \tilde{\rho} \ge \frac{1}{2} \implies \exists \, \omega_c \in (-\infty, \infty), \; 0 \in W(\tilde{S}(\omega_c)). 
\end{equation}

\begin{figure}[htbp]
    \centering
    \includegraphics[width=0.36\textwidth]{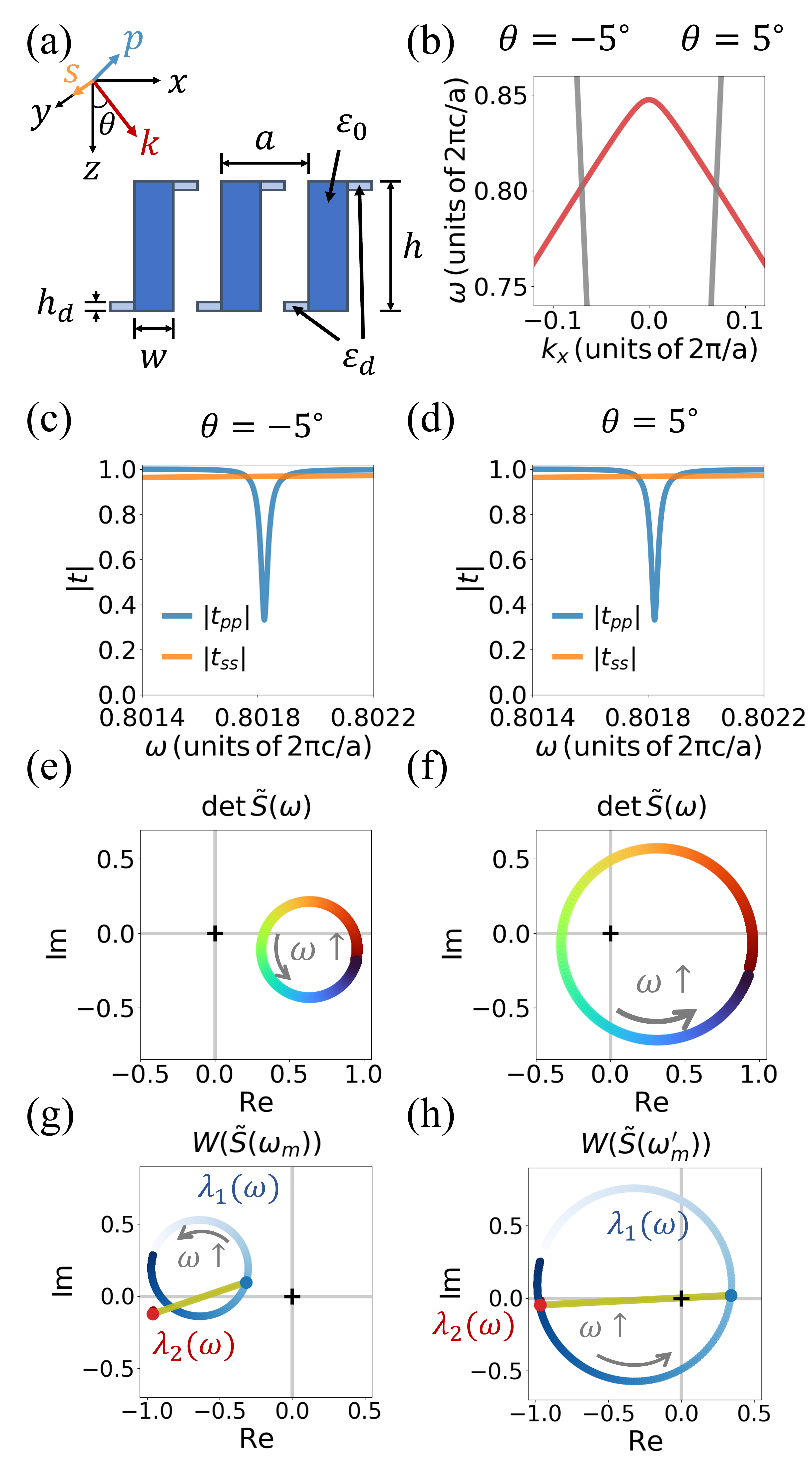}
    \caption{\red{A grating example. (a) Geometry. (b) Band structure. (c,e,g) Results for $\theta= -\ang{5}$. (c) $|t_{pp}|(\omega)$ and $|t_{ss}|(\omega)$.  (e) $\det \tilde{S}(\omega)$. (g) Trajectories of $\lambda_1(\omega)$ and $\lambda_2(\omega)$. The red and blue dots denote   $\lambda_1(\omega_m)$ and $\lambda_2(\omega_m)$, respectively. The line segment denotes $W(\tilde{S})$ at $\omega_m = 0.80182\times 2\pi c /a$. (d,f,h) Corresponding results for $\theta= \ang{5}$. In (h), $\omega'_m = 0.80220\times 2\pi c /a$.} }
    \label{fig:demo_structure}
\end{figure}

\red{We illustrate our theory with two numerical examples. In both cases, coherent orthogonal scattering corresponds to complete polarization conversion in the transmission. \rred{The effect of coherent orthogonal scattering can also be demonstrated in other (e.g., spatial) degrees of freedom. Our theoretical analysis applies to these cases as well.}}

\red{The first example is a 1D dielectric grating taken from Ref.~\cite{zhou2016}. The structure has a periodicity $a$. Each unit cell has a central rod of width $w=0.45 a$ and height $h=1.5 a$ with a dielectric constant $\varepsilon_0 = 2.1025$, and two additional pieces of width $(a-w)/2 = 0.275 a$ and height $h_d = 0.1 a$ with a dielectric constant $\varepsilon_d = 1.21$ (Fig.~\ref{fig:demo_structure}a). This grating supports a band of $p$-polarized guided resonances near $\omega = 0.8 \times 2\pi c /a$ (Fig.~\ref{fig:demo_structure}b). We consider an incident plane wave with an incident angle $\theta=\pm 5^\circ$, and calculate the $2\times 2$ transmission matrix}
\begin{equation}\label{eq:S_2by2_special}
\tilde{S}(\omega) = \begin{pmatrix}
t_{pp} & 0  \\
0 & t_{ss}
\end{pmatrix},
\end{equation}
\red{which is diagonal due to the $xz$ mirror symmetry. Consequently, we have $\lambda_1 = t_{pp}$, $\lambda_2 = t_{ss}$, $\det \tilde{S} = t_{pp} \, t_{ss}$, and $W(\tilde{S})$ is the line segment connecting $t_{pp} $ and $t_{ss}$.}

\red{Figs.~\ref{fig:demo_structure}(c,e,g) present the results for $\theta=\ang{-5}$. Fig.~\ref{fig:demo_structure}c shows that $|t_{pp}(\omega)|$ exhibits a single guided resonance~\cite{Fan2002} while $|t_{ss}(\omega)|$ almost remains constant. We fit the complex spectra with Eq.~(\ref{eq:S_tilde_sub_matrix}) and obtain $\tilde{\rho} = 0.334 - 0.009 i$. Since $\operatorname{Re}\tilde{\rho}<\frac{1}{2}$, Eq.~(\ref{eq:main_result}) predicts that $\operatorname{wn}(\tilde{S}) = 0$. Indeed, Fig.~\ref{fig:demo_structure}e shows that $\det \tilde{S}(\omega)$ traces out a circle that does not enclose the origin. Fig.~\ref{fig:demo_structure}g shows that $\lambda_1(\omega)$ forms a circle that does not enclose the origin, while $\lambda_2(\omega)$ almost remains constant. Fig.~\ref{fig:demo_structure}g also depicts $W(\tilde{S})$ at $\omega_m = 0.80182\times 2\pi c /a$, where $c(\tilde{S})$ reaches minimum. $W(\tilde{S}(\omega_m))$ forms a line segment connecting $\lambda_1(\omega_m)$ and $\lambda_2(\omega_m)$ with $0 \notin W(\tilde{S}(\omega_m))$. Thus, coherent orthogonal scattering does not occur in the entire frequency range.}

\red{Figs.~\ref{fig:demo_structure}(d,f,h) present the results for $\theta=\ang{5}$. The transmission amplitude spectra in Fig.~\ref{fig:demo_structure}d are similar to those in Fig.\ref{fig:demo_structure}c; the differences are in the phases. We fit the complex spectra and obtain $\tilde{\rho} = 0.666 - 0.009 i$. Since $\operatorname{Re}\tilde{\rho} > \frac{1}{2}$, Eq.~(\ref{eq:main_result}) states that $\operatorname{wn}(\tilde{S}) = 1$, and (\ref{eq:condition_resonance}) predicts that coherent orthogonal scattering must occur at some frequency. Indeed, Fig.~\ref{fig:demo_structure}\rred{f} shows that $\det \tilde{S}(\omega)$ traces out a circle that winds around the origin once counterclockwise. Fig.~\ref{fig:demo_structure}\rred{h} shows that $\lambda_1(\omega)$ winds around the origin once, while $\lambda_2(\omega)$ almost remains constant. Fig.~\ref{fig:demo_structure}\rred{h} depicts $W(\tilde{S})$ at $\omega'_m = 0.80220\times 2\pi c /a$, which is the line segment connecting $\lambda_1(\omega'_m)$ and $\lambda_2(\omega'_m)$ with $0 \in W(\tilde{S}(\omega'_m))$. Explicit calculation shows that coherent orthogonal scattering occurs at $\omega'_m$ for the incident polarization $(0.86\hat{p} + 0.51 \mathrm{e}^{i\rred{\eta}}\hat{s})$ with $\rred{\eta}$ an arbitrary phase.}

\begin{figure}[htbp]
    \centering
    \includegraphics[width=0.36\textwidth]{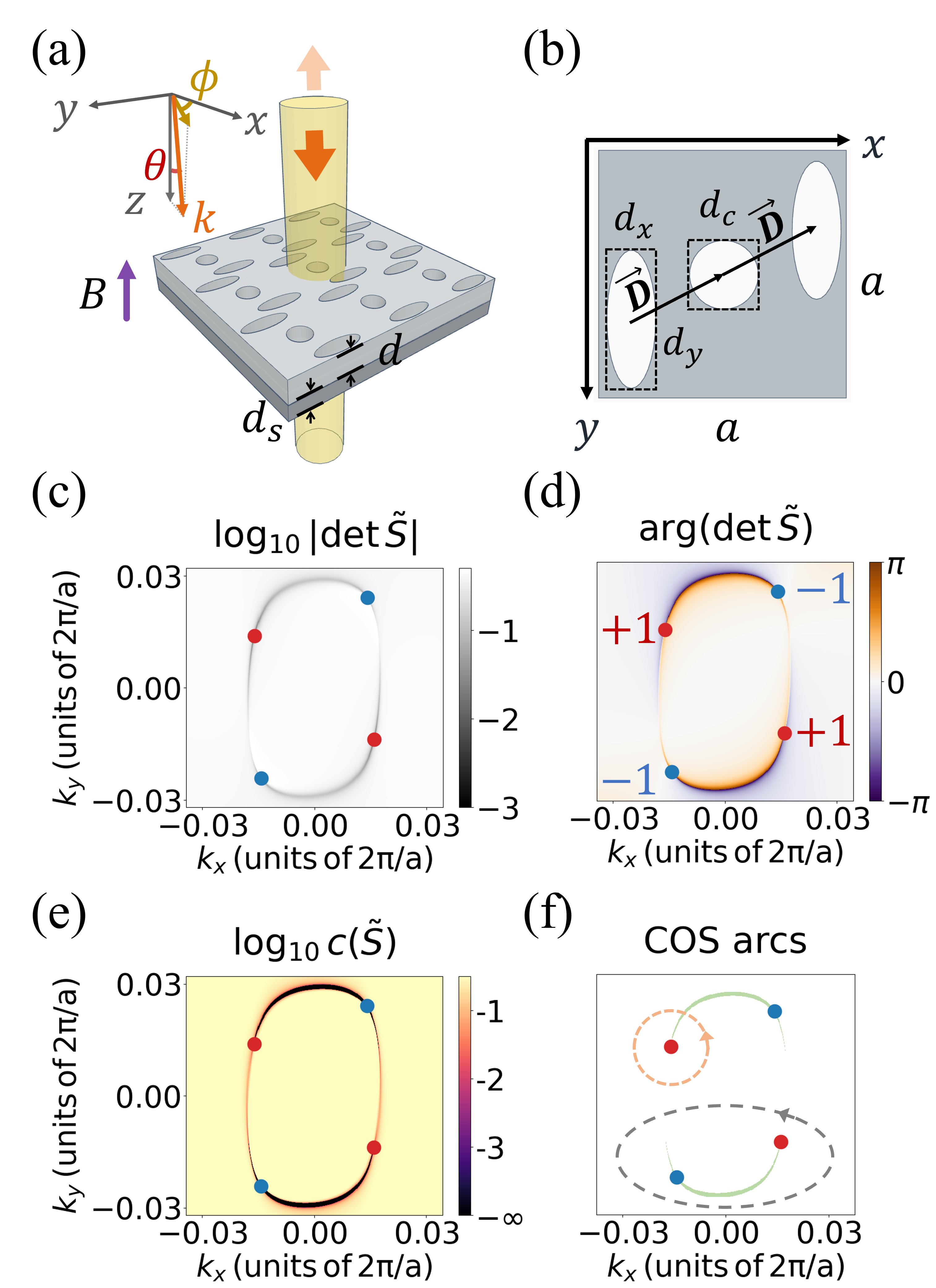}
    \caption{\red{A photonic crystal slab example. (a) Geometry. (b) Unit cell.} (c-f) Results for $\tilde{S}(k_x,k_y)$ at $\omega=0.386\times 2\pi c /a$. (\rred{c}) $\log_{10} \lvert \det \tilde{S}(k_x,k_y)\rvert$.  (\rred{d}) $\arg [\det \tilde{S} (k_x,k_y)]  $. (\rred{e}) $\log_{10}  c[\tilde{S}  (k_x,k_y)]$. (\rred{f}) The green region: the COS arcs. The red and blue dots: CPE points with winding numbers $+1$ and $-1$, respectively.
    }
    \label{fig:S_kx_ky}
\end{figure}

\red{The second example is a 2D magneto-optical photonic crystal slab.} The structure consists of two layers (Fig.~\ref{fig:S_kx_ky}a). The first layer is a photonic crystal slab with a lattice constant $a$ and a thickness $d=\qty{0.3}{\text{$a$}}$. It has a square lattice, with each unit cell having a circular hole at the center with a diameter $d_{c} = \qty{0.28}{\text{$a$}}$, and two elliptical holes displaced by $\pm \vec{D} = \pm (\qty{0.375}{\text{$a$}}, \qty{-0.18}{\text{$a$}})$ from the center, where the major axis is $d_{y} = \qty{0.56}{\text{$a$}}$ and the minor axis is $d_{x} = \qty{0.20}{\text{$a$}}$ (Fig.~\ref{fig:S_kx_ky}b). The second layer is a uniform slab with a thickness $d_{s} = \qty{0.2}{\text{$a$}}$. Both slabs are made of a magneto-optical material with a relative permittivity tensor
\begin{equation}
\varepsilon = \begin{pmatrix}
\varepsilon_r  &  i \varepsilon_a & 0  \\
-i \varepsilon_a & \varepsilon_r & 0  \\
0 & 0 & \varepsilon_r 
\end{pmatrix},
\end{equation}
where $\varepsilon_{r} = 12.0$, $\varepsilon_a = 5.0$. \red{We consider a plane wave with a fixed frequency $\omega=0.386\times 2\pi c /a$ incident from the top side, and calculate the $2\times 2$ transmission matrix}
\begin{equation}\label{eq:S_2by2_general}
\tilde{S}(k_x,k_y) = \begin{pmatrix}
t_{xx} & t_{xy}  \\
t_{yx} & t_{yy}
\end{pmatrix},
\end{equation}
\red{where $(k_x,k_y)$ are the transverse wavevector components. Unlike that in Eq.~(\ref{eq:S_2by2_special}), here $\tilde{S}$ is a generic $2 \times 2$ matrix. In general, $W(\tilde{S})$ forms an elliptical disk with foci at $\lambda_1$ and $\lambda_2$ (see SM, Sec.~\ref{SI-subsec:geometry}.)}

Fig.~\ref{fig:S_kx_ky}c plots 
\begin{equation}
    \log_{10}  |\det \tilde{S} (k_x, k_y)|,
\end{equation}
which exhibits sharp variations near the isofrequency contour of the guided resonances. There are four CPE points where $\det \tilde{S} = 0$. Fig.~\ref{fig:S_kx_ky}d shows that $\det \tilde{S}(k_x, k_y)$ exhibits a $2\pi$ ($-2\pi$) phase winding around each red (blue) CPE point, indicating a winding number of $+1$ ($-1$). The two red (blue) CPE points are related by the $C_2$ symmetry. Fig.~\ref{fig:S_kx_ky}e and~\ref{fig:S_kx_ky}f  show that $c[S(k_x,k_y)]=0$ occurs in two disjoint COS arcs. Each COS arc occupies a two-dimensional domain that contains one red and one blue CPE point in its interior. These COS arcs are topologically protected. For instance, the orange loop encircles a single CPE point, so there must exist a COS arc that intersects the orange loop. In contrast, the grey loop encircles two CPE points with a total winding number $N_+ - N_- = 1-1 = 0$. Since $2 \mid 0$, a COS arc may or may not cross the grey loop. Indeed, in this case, the COS arc is contained entirely inside the loop. \red{(See SM, Sec.~\ref{SI-sec:additional_numerical_results} for additional numerical results.)} These numerical examples confirm our theoretical predictions.

In conclusion, we have introduced the phenomenon of coherent orthogonal scattering\rred{, which occurs when the output wave becomes orthogonal to the reference output state in the absence of scatterers. This effect leads to a unity extinction coefficient and complete mode conversion.} We have revealed the topological nature of this effect by establishing a connection between the winding number of scattering submatrices and the existence of coherent orthogonal scattering. We have discovered topologically protected coherent orthogonal scattering arcs in a two-dimensional parameter space, which connect coherent perfect extinction points with opposite winding numbers. These findings provide a deeper understanding of the interplay between topology and scattering phenomena, paving the way for novel approaches to manipulate coherent wave interactions in various physical systems.

\begin{acknowledgments}
This work is funded by a Simons Investigator in Physics
grant from the Simons Foundation (Grant No. 827065),
by a Multidisciplinary University Research Initiative
(MURI) grant from the U. S. Office of Naval Research (Grant No. N00014-20-1-2450), and by U. S. Department of Energy (Grant No. DE-FG02-07ER46426).
\end{acknowledgments}

\appendix 

\section{Numerical range}\label{SI-sec:numerical_range}

Here we briefly review the numerical range. We refer readers to  Refs.~\cite{horn2008,gustafson1997} for a more detailed introduction.  

\subsection{Definition}\label{SI-subsec:numerical_range_definition}

\begin{definition}[Numerical range]
Let $A \in M_{n}$. The numerical range of $A$ is the subset of $\mathbb{C}$ defined as 
\begin{equation}
W(A) \coloneqq \{x^{\dagger} A x: x \in \mathbb{C}^{n}, x^{\dagger}x = 1\}.    
\end{equation}
\end{definition}
\begin{remark}
The numerical range is also known as the field of values in the literature~\cite{horn2008,gustafson1997}. 
\end{remark}

\subsection{Properties}\label{SI-subsec:properties}

The numerical range has many nice properties. 

\begin{notation}
Let $D \subseteq \mathbb{C}$. $D$ is compact if it is closed and bounded; $D$ is convex if a line segment $L$ joining two points in $D$ satisfies $L \subseteq D$~\cite{munkres2000}.    
\end{notation}

\begin{property}[Compactness]
For all $A\in M_{n}$, $W(A)$ is a compact subset of $\mathbb{C}$.
\end{property} 
\begin{proof}
$W(A)$ is the range of the unit sphere $\{x\in \mathbb{C}^n: x^{\dagger}x=1\}$ under the continuous map $x \mapsto x^{\dagger}A x$. Since the unit sphere is a compact set and the continuous image of a compact set is compact~\cite{munkres2000}, $W(A)$ is compact.
\end{proof}
\begin{property}[Convexity]\label{SI-property:convexity} For all $A \in M_{n}$, $W(A)$ is a convex subset of $\mathbb{C}$.  
\end{property}
\begin{proof}
See Ref.~\cite{horn2008} pp.~17-20, Sec.~1.3. One first reduces the problem to the $2\times 2$ case, then proves the convexity in the $2\times 2$ case.     
\end{proof}
\begin{remark}
The convexity of the numerical range is known as the \emph{Toeplitz-Hausdorff theorem}. Toeplitz~\cite{toeplitz1918} showed that the outer boundary curve of $W(A)$ is convex; Hausdorff~\cite{hausdorff1919} showed that $W(A)$ is simply connected.    
\end{remark}

\begin{property}[Spectral containment]\label{SI-property:spectral_containment} For all $A \in M_{n}$, 
\begin{equation}
\bm{\lambda}(A) \subseteq W(A). 
\end{equation}    
\end{property}
\begin{proof}
Suppose that $\lambda \in \bm{\lambda}(A)$. Then there exists a unit vector $x \in \mathbb{C}^n$ for which $Ax = \lambda x$ and hence $\lambda = \lambda x^{\dagger} x =  x^{\dagger} (\lambda x) = x^{\dagger} A x \in W(A)$.     
\end{proof}

\subsection{Geometry}\label{SI-subsec:geometry}
We discuss two questions about the geometric shapes of the numerical ranges. 
\begin{question}\label{question:Kippenhahan}
Which subsets of $\mathbb{C}$ occur as the numerical range of some $n \times n$ matrices? 
\end{question}

\begin{question}\label{question:determine_WA}
Given $A \in M_{n}$, how to determine the shape of $W(A)$ from the entries of $A$?    
\end{question}

Question~\ref{question:Kippenhahan} was originally raised by Kippenhahn in 1951~\cite{kippenhahn1951a} (see also the English translation~\cite{zachlin2008}). In that seminal paper, Kippenhahn characterized the numerical range of a matrix as being the convex hull of a certain algebraic curve associated with the matrix. This ``boundary generating curve" has led to many useful results and remained a topic of current research~\cite{helton2007,henrion2010,helton2011}. Question~\ref{question:determine_WA} is a harder question and remains open. 

When $n=2$, the answers to both questions are known:
\begin{theorem}\label{theorem:2_by_2_numerical_range}
Suppose $A \in M_{2}$ has eigenvalues $\lambda_{1}$, $\lambda_{2}$. Then $W(A)$ is an elliptical disk with foci $\lambda_{1}$, $\lambda_{2}$ and minor axis length $\sqrt{\mathrm{Tr}(A^{\dagger}A) - \lvert \lambda_{1} \rvert^2 - \lvert \lambda_{2} \rvert^2}$. 
\end{theorem}
\begin{proof}
See Ref.~\cite{horn2008}, pp.~20, Lemma~1.3.3.  
\end{proof}

When $n=3$, the answers are known but more complicated. Question~\ref{question:Kippenhahan} was answered by Kippenhahn~\cite{kippenhahn1951a}, who showed that the shape of $W(A)$ with $A \in M_{3}$ can be a (possibly degenerate) triangle, an ellipse, a cone-like shape, an ovular shape, or a shape with a flat portion on the boundary. Question~\ref{question:determine_WA} was answered by Keeler et al.~in 1997~\cite{keeler1997}, who provided a series of tests to determine the shape of $W(A)$ from the entries of $A \in M_{3}$. 

When $n>3$, the answers are generally unknown except in some special cases~\cite{brown2004,gau2006,militzer2017,cox2020,geryba2021}. Ref.~\cite{helton2011} provides some useful criteria that address Question~\ref{question:Kippenhahan} in a certain sense.

\subsection{Numerical algorithm}\label{SI-subsec:numerical_algorithm}

Although an analytical characterization of $W(A)$ for arbitrary $A\in M_{n}$ is lacking, a numerical procedure exists for determining and plotting $W(A)$. As $W(A)$ is convex and compact, it is sufficient to determine its boundary, $\partial W(A)$. The approach involves calculating many well-spaced points on $\partial W(A)$ and identifying the support lines of $W(A)$ at these points. The convex hull of these boundary points provides an \emph{internal} convex polygonal approximation to $W(A)$ that is contained within $W(A)$, while the intersection of the half-spaces determined by the support lines provides an \emph{external} convex polygonal approximation to $W(A)$ that contains $W(A)$. A detailed procedure is available in Ref.~\cite{horn2008}, pp.~33-39.

\section{Examples of scattering submatrices}\label{SI-sec:examples_scattering_submatrices}

Here we provide examples of scattering submatrices that  exhibit coherent orthogonal scattering or not:
\begin{equation}\label{eq:example_S}
\tilde{S} = \begin{cases}
0, &m=1\\ 
\sigma_{x}, &m=2 \\
\sigma_{x} \oplus I_{m-2} &m\geq 3
\end{cases}  
\end{equation}
exhibits coherent orthogonal scattering  because \begin{equation}
\tilde{\bm{a}} = (1,0,\dots,0)^T \in \mathbb{C}^m 
\end{equation}
is a solution to Eq.~(\ref{eq:COS_definition}). Here $\oplus$ denotes the direct sum of matrices,
\begin{equation}
\sigma_{x} = \begin{pmatrix}
0 & 1  \\
1 & 0
\end{pmatrix},   
\end{equation}
and $I_{m}$ denotes the $m \times m$ identity matrix. In contrast,
\begin{equation}
\tilde{S}' =   I_{m}, \quad m\geq 1,   \label{eq:example_S_prime}
\end{equation}
does not exhibit coherent orthogonal scattering, since 
\begin{equation}
\tilde{\bm{a}}^{\dagger} \tilde{S}' \tilde{\bm{a}} = |\tilde{\bm{a}}|^2 > 0, \quad \forall \tilde{\bm{a}} \neq \bm{0}.  
\end{equation}
Figs.~\ref{fig:numerical_range_example}a and~\ref{fig:numerical_range_example}b plot the numerical range of $\tilde{S}$ in Eq.~(\ref{eq:example_S}) and $\tilde{S}'$ in Eq.~(\ref{eq:example_S_prime}), respectively. We see that $0 \in W(\tilde{S})$ and $0 \notin W(\tilde{S}')$, which confirms condition~(\ref{eq:inclusion_relation}). 

\begin{figure}[htbp]
    \centering
\includegraphics[width=0.5\textwidth]{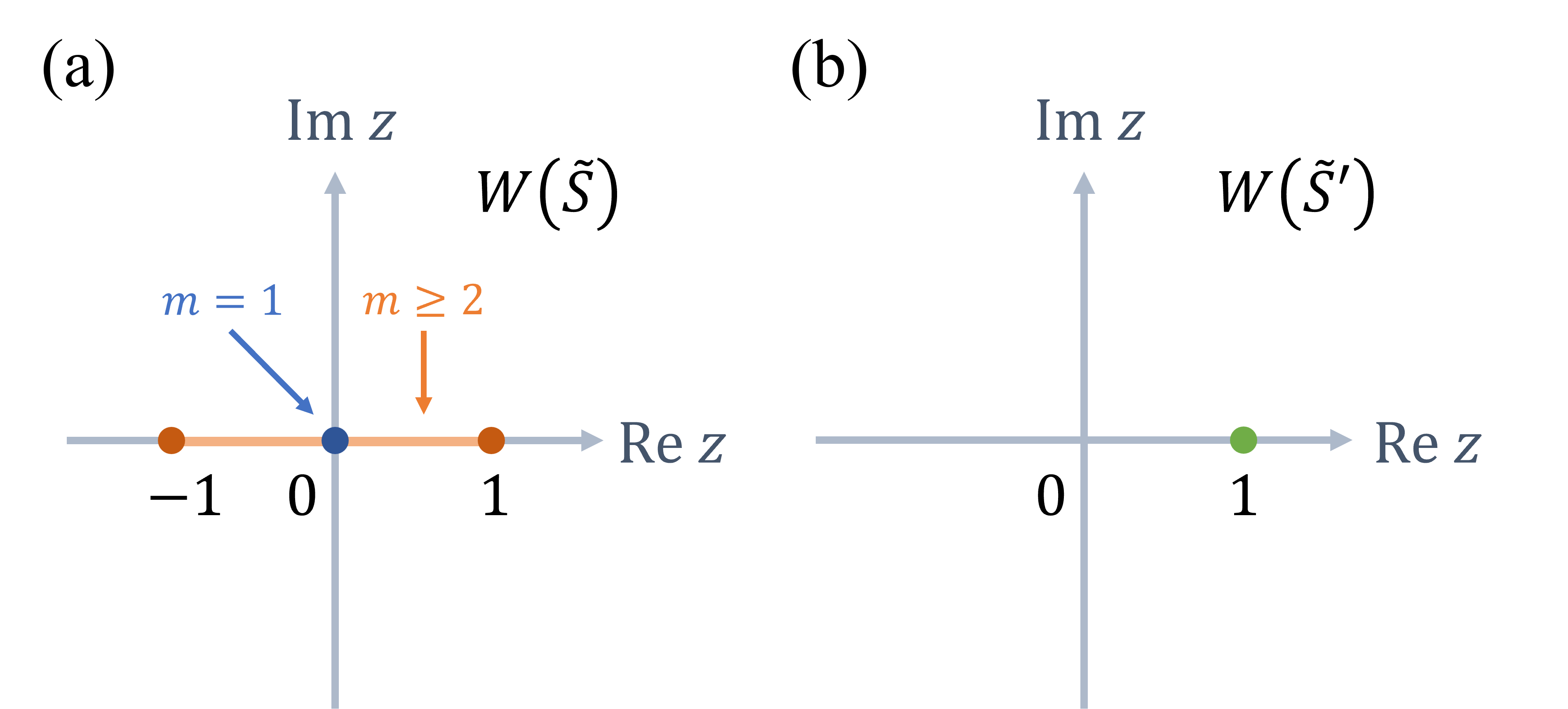}
    \caption{Numerical range of (a) $\tilde{S}$ in Eq.~(\ref{eq:example_S}); (b) $\tilde{S}'$ in Eq.~(\ref{eq:example_S_prime}). }
    \label{fig:numerical_range_example}
\end{figure}

Now we show that our criterion~(\ref{eq:criterion_to_prove}) is sufficient but not necessary: If 
\begin{equation}\label{eq:m_divides_wind_S_SI}
m \mid \operatorname{wn(\tilde{S})},    
\end{equation} 
there may or may not exist $\tau_{0} \in [0,1]$ such that $0 \in W(\tilde{S}(\tau_0))$. For example, consider
\begin{equation}
    \tilde{S}_{k}(\tau) = e^{i2\pi k\tau} \tilde{S}, \quad \tilde{S}_{k}'(\tau) = e^{i2\pi k\tau} \tilde{S}', \quad k \in \mathbb{Z}.    
\end{equation}
where $\tilde{S}$ and $\tilde{S}'$ are defined in Eqs.~(\ref{eq:example_S}) and (\ref{eq:example_S_prime}), respectively. One can show that 
\begin{equation}
\operatorname{wn}(\tilde{S}_{k}) = \operatorname{wn}(\tilde{S}'_{k}) = mk.    
\end{equation}
Hence 
\begin{equation}
m \mid  \operatorname{wn}(\tilde{S}_{k}), \quad m \mid  \operatorname{wn}(\tilde{S}'_{k}).  
\end{equation}
Nonetheless, for all $\tau \in [0,1]$, 
\begin{equation}
0 \in W(\tilde{S}_{k}(\tau)), \quad 0 \notin W(\tilde{S}'_{k}(\tau)).    
\end{equation}

\section{Kato's theorem}\label{SI-Sec:Kato}
\begin{theorem}[Kato, 1966]
Suppose that $D$ is a real interval and that $A: D \to M_{n}$ is a continuous function. Then there exist $n$ eigenvalues (counted with algebraic multiplicities) of $A(t)$ that can be parameterized as continuous functions $\lambda_{1}(t), \dots, \lambda_{n}(t)$ from $D$ to $\mathbb{C}$.
\end{theorem}
\begin{proof}
See Ref.~\cite{kato1995},~pp. 109,~Theorem 5.2. 
\end{proof}
\begin{remark}
Kato's result is known as the functional continuity of eigenvalues~\cite{li2019a}. It sounds simple but is actually tricky. For example, the statement no longer holds if $D$ is a domain with interior points in the complex plane. See detailed discussions in Ref.~\cite{li2019a}.
\end{remark}

\section{Matrix determinant lemma}\label{SI-sec:matrix_determinant_lemma}
\begin{lemma*}
Suppose $A$ is an $n\times n$ invertible matrix, $\bm{u}$ and $\bm{v}$ are $n\times 1$ column vectors, then 
\begin{equation}\label{eq:lemma_det_general_case}
\det \left(A+\bm{u} \bm{v}^T\right) = (1+\bm{v}^T A^{-1} \bm{u})  \, \det A.     
\end{equation}
\end{lemma*}
\begin{proof}
The proof is taken from Ref.~\cite{meyer2000}, pp.~475. First, we prove the special case $A = I$. We use the following block matrix equality:
\begin{equation}
\begin{pmatrix}
I & 0  \\
\bm{v}^T & 1
\end{pmatrix}
\begin{pmatrix}
I + \bm{u}\bm{v}^T & \bm{u} \\
0 & 1
\end{pmatrix}
\begin{pmatrix}
I & 0  \\
-\bm{v}^T & 1
\end{pmatrix}
= \begin{pmatrix}
I & \bm{u} \\
0 & 1+ \bm{v}^T \bm{u}
\end{pmatrix}.    
\end{equation}
We take the determinant on both sides and use the product rule of determinant to obtain 
\begin{equation}\label{eq:lemma_det_special_case}
\det (I+\bm{u}\bm{v}^T) = 1+\bm{v}^T \bm{u}.     
\end{equation}
To prove the general case, we write 
\begin{equation}
A+\bm{u}\bm{v}^T = A(I+A^{-1}\bm{u} \bm{v}^T).    
\end{equation}
We take the determinant on both sides and use Eq.~(\ref{eq:lemma_det_special_case}) to obtain Eq.~(\ref{eq:lemma_det_general_case}).    
\end{proof}

\section{Proof of Eq.~(\ref{eq:main_result})}\label{SI-sec:proof_resonance_winding}

Here we provide a detailed proof of Eq.~(\ref{eq:main_result}). We apply the matrix determinant lemma to Eq.~(\ref{eq:S_tilde_sub_matrix}) and obtain 
\begin{equation}
\det \tilde{S}(\omega) = \tilde{f}(\Omega) \det \tilde{C}.
\end{equation}
where 
\begin{gather}
\tilde{f}(\Omega) \coloneqq \frac{-i \Omega + (1-2\tilde{\rho})}{-i \Omega + 1}, \\
\Omega \coloneqq \frac{\omega - \omega_{0}}{\gamma}, \quad \tilde{\rho} \coloneqq -\frac{\tilde{\bm{\kappa}}^T \tilde{C}^{-1} \tilde{\bm{d}}}{2 \gamma} \in \mathbb{C}.    
\end{gather}
Now we calculate the winding number
\begin{equation}\label{eq:def_wind_S_resonance}
\operatorname{wn}(\tilde{S}) = \frac{1}{2\pi i} \int_{-\infty}^{\infty} \, \frac{1}{\det \tilde{S}} \frac{\mathrm{d}\det \tilde{S} }{\mathrm{d}\omega} \mathrm{d}\omega.
\end{equation}
If $\det \tilde{C} = 0$, $\operatorname{wn}(\tilde{S})$ is undefined. If $\det \tilde{C} \neq 0$, 
\begin{equation}
\operatorname{wn} (\tilde{S}) = \operatorname{wn} (\tilde{f}).  
\end{equation}
$\tilde{f}(\Omega)$ is a M\"{o}bius transformation~\cite{needham2021}, which maps the real line into a circle centered at $1-\tilde{\rho}$ with a radius of $|\tilde{\rho}|$ (Fig.~\ref{fig:phase_diagram}a). The circle encloses the origin if and only if $\operatorname{Re} \tilde{\rho} > \frac{1}{2}$. 
As $\Omega$ runs over $[-\infty,+\infty]$, $\tilde{f}(\Omega)$ traces out the circle once counterclockwise. Hence, we obtain
\begin{equation}\label{eq:main_result_SI}
\operatorname{wn}(\tilde{S}) = \begin{cases}
0,  & \operatorname{Re} \tilde{\rho} < \frac{1}{2} \\
\textup{Undefined}, & \operatorname{Re} \tilde{\rho} = \frac{1}{2} \\ 
1, & \operatorname{Re} \tilde{\rho} > \frac{1}{2}
\end{cases}    
\end{equation}
which completes the proof of Eq.~(\ref{eq:main_result}). 
\begin{figure}[htbp]
    \centering
\includegraphics[width=0.5\textwidth]{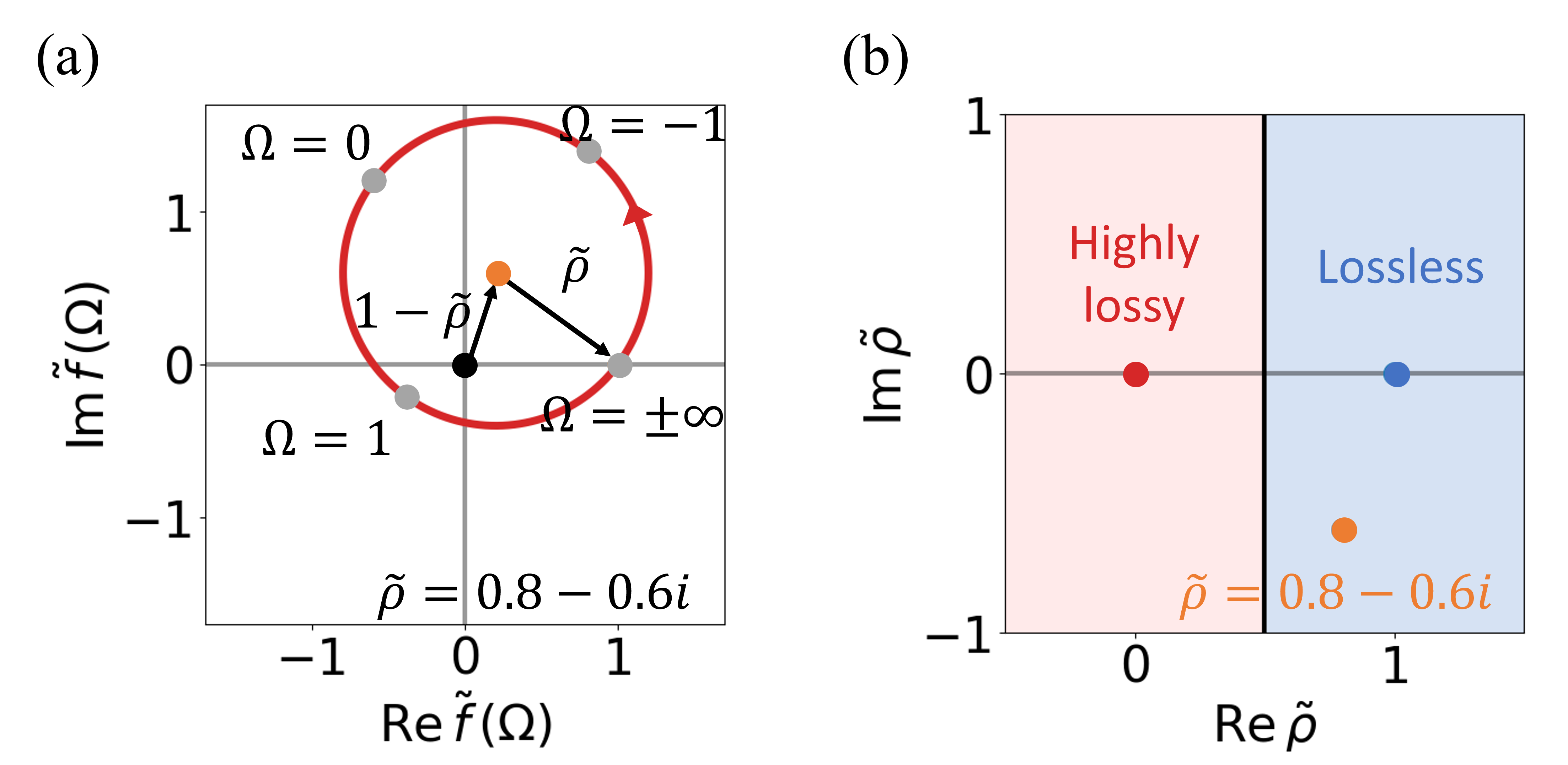}
    \caption{ Topological winding of a single resonance. (c) The path of $\tilde{f}(\Omega)$. (d) The topological phase diagram. }
    \label{fig:phase_diagram}
\end{figure}
\section{Topological phases of a resonance}\label{SI-sec:topological_phases_resonance}

Eq.~(\ref{eq:main_result}) indicates two topological phases for a resonance, one with $\operatorname{wn}(\tilde{S}) = 0$ and the other with $\operatorname{wn}(\tilde{S}) = 1$ (Fig.~\ref{fig:phase_diagram}b).  We illustrate them with two special cases.

\begin{enumerate}
    \item \textbf{Lossless case.} For a lossless system with a \emph{total} scattering matrix $S$~\cite{zhao2019c},  
    \begin{equation}
    \det C \neq 0, \qquad C \bm{\kappa}^* + \bm{d}  = \bm{0},  \qquad   \bm{\kappa}^{\dagger} \bm{\kappa} = 2 \gamma.  
    \end{equation}

Thus
\begin{equation}
\rho = -\frac{\bm{\kappa}^T C^{-1} \bm{d}}{2 \gamma} = \frac{\bm{\kappa}^T \bm{\kappa}^*}{2 \gamma} = \frac{(\bm{\kappa} ^{\dagger} \bm{\kappa})^*}{2 \gamma} = 1.  
\end{equation}
The image of $f(\Omega)$ is the unit circle centered at the origin. Hence $\operatorname{wn}(S) = 1$.
\item \textbf{Highly lossy case.} For a highly lossy system with any scattering submatrix $\tilde{S}$, $\gamma\to +\infty$, thus
\begin{equation}
\tilde{\rho} = -\frac{\tilde{\bm{\kappa}}^T \tilde{C}^{-1} \tilde{\bm{d}}}{2 \gamma} \to 0, \qquad \tilde{S}(\omega) \to \tilde{C}.  
\end{equation}
The image of $\tilde{f}(\Omega)$ shrinks to a single point $1$ in the complex plane. Hence $\operatorname{wn}(\tilde{S}) = 0$.
\end{enumerate}

\section{Additional Numerical Results}\label{SI-sec:additional_numerical_results}

\begin{figure}[htbp]
    \centering
    \includegraphics[width=0.36\textwidth]{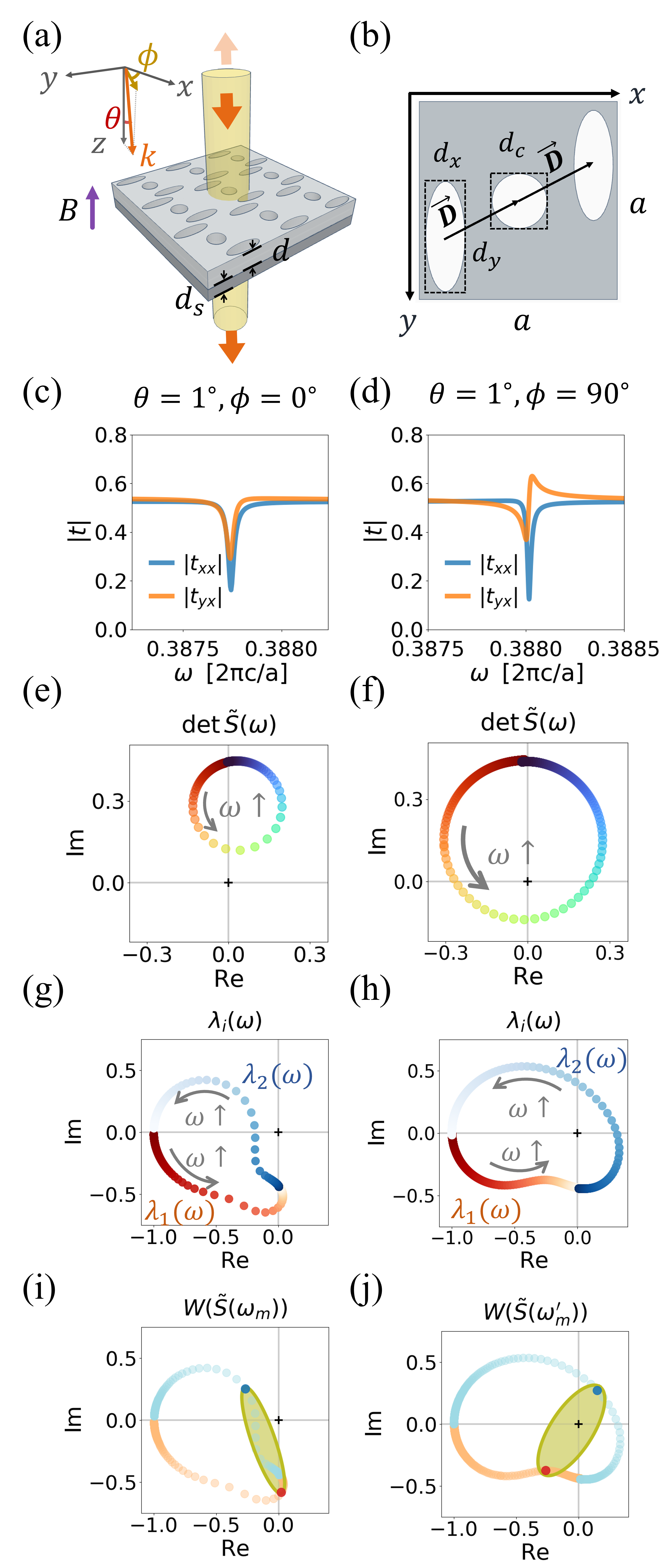}
    \caption{Additional numerical results for the photonic crystal slab example. (a) Geometry. (b) Unit cell. (c,e,g,i) Results for $\tilde{S}(\omega;\theta= \ang{1}, \phi = \ang{0})$. (c) $|t_{xx}|(\omega)$ and $|t_{yx}|(\omega)$.  (e) $\det \tilde{S}(\omega)$. (g) $\lambda_1(\omega)$ and $\lambda_2(\omega)$. (i) The ellipse denotes $W(\tilde{S})$ at $\omega_m = 0.3877\times 2\pi c /a$. The red and blue dots denote   $\lambda_1(\omega_m)$ and $\lambda_2(\omega_m)$, respectively. Also shown are the trajectories of $\lambda_1(\omega)$ and $\lambda_2(\omega)$. (d,f,h,j) Corresponding results for $\tilde{S}(\omega;\theta= \ang{1}, \phi = \ang{90})$. In (j), $\omega'_m = 0.3880\times 2\pi c /a$. }
    \label{fig-SI:demo_structure}
\end{figure}

\red{Here we provide additional numerical results for the photonic crystal slab example. For readers' reference, the structure is reproduced in Figs.~\ref{fig-SI:demo_structure}(a,b). }

\red{Figs.~\ref{fig-SI:demo_structure}(c,e,g,i) present the results for $\tilde{S}(\omega;\theta=\ang{1},\phi=\ang{0})$. Fig.~\ref{fig-SI:demo_structure}c shows that the transmission spectra exhibit a single guided resonance~\cite{Fan2002}. We fit the spectra with Eq.~(\ref{eq:S_tilde_sub_matrix}) and obtain $\tilde{\rho} = 0.361 + 0.086 i$. Since $\operatorname{Re}\tilde{\rho}<\frac{1}{2}$, Eq.~(\ref{eq:main_result}) states that $\operatorname{wn}(\tilde{S}) = 0$. We now verify these predictions.  Fig.~\ref{fig-SI:demo_structure}e shows that $\det \tilde{S}(\omega)$ traces out a circle that does not enclose the origin. Fig.~\ref{fig-SI:demo_structure}g shows that as $\omega$ runs over the entire range, the eigenvalues $\lambda_1$ and $\lambda_2$ are permutated, but $[\arg (\lambda_1) + \arg (\lambda_2)]$ does not change. These observations confirm that $\operatorname{wn}(\tilde{S}) = 0$. Fig.~\ref{fig-SI:demo_structure}i depicts $W(\tilde{S})$ at $\omega_m = 0.3877\times 2\pi c /a$, where $c(\tilde{S})$ is minimized. $W(\tilde{S}(\omega_m))$ forms an elliptical disk with foci at $\lambda_1(\omega_m)$ and $\lambda_2(\omega_m)$ (see SM, Sec.~\ref{SI-subsec:geometry}.). Since $0 \notin W(\tilde{S}(\omega_m))$, coherent orthogonal scattering does not occur in the entire frequency range.}

\red{Figs.~\ref{fig-SI:demo_structure}(d,f,h,j) present the results for $\tilde{S}(\omega;\theta=\ang{1},\phi=\ang{90})$. Fig.~\ref{fig-SI:demo_structure}d shows that the transmission spectra also exhibit a single guided resonance. We fit the spectra with Eq.~(\ref{eq:S_tilde_sub_matrix}) and obtain $\tilde{\rho} = 0.657 - 0.027 i$. As $\operatorname{Re}\tilde{\rho} > \frac{1}{2}$, Eq.~(\ref{eq:main_result}) states that $\operatorname{wn}(\tilde{S}) = 1$. Consequently, (\ref{eq:condition_resonance}) states that coherent orthogonal scattering occurs at some frequency. We now verify these predictions. Fig.~\ref{fig-SI:demo_structure}f shows that $\det \tilde{S}(\omega)$ traces out a circle that winds around the origin once counterclockwise. Fig.~\ref{fig-SI:demo_structure}h shows that as $\omega$ runs over the entire range, the eigenvalues $\lambda_1$ and $\lambda_2$ are permutated, and $[\arg (\lambda_1) + \arg (\lambda_2)]$ increases by $2\pi$. These observations confirm that $\operatorname{wn}(\tilde{S}) = 1$. Fig.~\ref{fig-SI:demo_structure}j depicts $W(\tilde{S})$ at $\omega'_m = 0.3880\times 2\pi c /a$. Again, $W(\tilde{S}(\omega'_m))$ forms an elliptical disk with foci at $\lambda_1(\omega'_m)$ and $\lambda_2(\omega'_m)$. However, now $0 \in W(\tilde{S}(\omega'_m))$, thus coherent orthogonal scattering occurs at $\omega'_m$.}


\bibliography{main}

\end{document}